\documentclass[a4paper,UKenglish,cleveref, autoref, thm-restate]{lipics-v2021}

\hideLIPIcs  

\title{A Compact DAG for Storing and Searching Maximal Common Subsequences} 

\author{Alessio Conte}{Università di Pisa, Italy}{alessio.conte@unipi.it}{https://orcid.org/0000-0003-0770-2235/}{}

\author{Roberto Grossi}{Università di Pisa, Italy}{roberto.grossi@unipi.it}{https://orcid.org/0000-0002-7985-4222/}{}

\author{Giulia Punzi}{National Institute of Informatics, Japan}{punzi@nii.ac.jp}{https://orcid.org/0000-0001-8738-1595/}{}

\author{Takeaki Uno}{National Institute of Informatics, Japan}{uno@nii.ac.jp}{https://orcid.org/0000-0001-7274-279X/}{}

\authorrunning{A. Conte, R. Grossi, G. Punzi, and T. Uno} 

\Copyright{Alessio Conte, Roberto Grossi, Giulia Punzi, and Takeaki Uno} 



\category{} 

\relatedversion{} 

\nolinenumbers 

\EventEditors{John Q. Open and Joan R. Access}
\EventNoEds{2}
\EventLongTitle{42nd Conference on Very Important Topics (CVIT 2016)}
\EventShortTitle{CVIT 2016}
\EventAcronym{CVIT}
\EventYear{2016}
\EventDate{December 24--27, 2016}
\EventLocation{Little Whinging, United Kingdom}
\EventLogo{}
\SeriesVolume{42}
\ArticleNo{23}

\usepackage{amsmath}
\usepackage{amssymb}
\usepackage{enumerate}
\usepackage{cite}
\usepackage{hyperref}
\usepackage{todonotes}
\usepackage{appendix}
\usepackage{thmtools}
\usepackage{thm-restate}
\usepackage{xspace}

\usepackage{algorithm}
\usepackage{algpseudocode}
\usepackage{caption} 
\usepackage{tikz}
\usetikzlibrary{decorations.pathreplacing, patterns, shapes.geometric}
\tikzstyle{vertex}=[circle, draw, inner sep=0pt, minimum size=6pt]
\tikzstyle{prefixvertex}=[circle, draw, fill = black!25!white, inner sep=0pt, minimum size=6pt]

\newcommand{\vertex}{\node[vertex]}
\newcommand{\prefixvertex}{\node[prefixvertex]}

\newcommand{\A}{\ensuremath\mathtt{A}}
\newcommand{\C}{\ensuremath\mathtt{C}}
\newcommand{\G}{\ensuremath\mathtt{G}}
\newcommand{\T}{\ensuremath\mathtt{T}}

\newcommand{\tswing}[1]{\ensuremath{ \ltimes_T(#1)}\xspace}
\newcommand{\bswing}[1]{\ensuremath {\ltimes_B(#1)}\xspace}
\newcommand{\recID}[4]{\ensuremath{\langle #1, #2, #3, #4 \rangle}\xspace}
\newcommand{\nodelab}[1]{\ensuremath{\mathit{ID}(#1)}\xspace}
\newcommand{\edgelab}[1]{\ensuremath{\ell(#1)}\xspace}
\newcommand{\stpaths}[3]{\ensuremath{\mathcal P_{#1, #2} (#3)}\xspace}

\newcommand{\MCSDAG}{\textsf{MDAG}\xspace}

\newcommand{\graphbuild}[1]{\ensuremath{\textsc{buildDAG}(#1)}\xspace}
\newcommand{\MCS}{\ensuremath{\mathit{MCS}\xspace}}
\newcommand{\Ext}{\ensuremath{\mathit{Ext}\xspace}}
\newcommand{\nextc}{\ensuremath{\mathit{next}\xspace}}

\begin{document}

\maketitle

\begin{abstract}
\emph{Maximal Common Subsequences} (MCSs) between two strings $X$ and $Y$ are subsequences of both $X$ and $Y$ that are maximal under inclusion. MCSs relax and generalize the well known and widely used concept of Longest Common Subsequences (LCSs), which can be seen as MCSs of maximum length. 
While the number both LCSs and MCSs can be exponential in the length of the strings, LCSs have been long exploited for string and text analysis, as simple compact representations of all LCSs between two strings, built via dynamic programming or automata, have been known since the '70s.
MCSs appear to have a more challenging structure: even listing them efficiently was an open problem open until recently, thus narrowing  the complexity difference between the two problems, but the gap remained significant. 
In this paper we close the complexity gap: we show how to build DAG of polynomial size---in polynomial time---which allows for efficient operations on the set of all MCSs such as enumeration in Constant Amortized Time per solution (CAT), counting, and random access to the $i$-th element (i.e., rank and select operations). Other than improving known algorithmic results, this work paves the way for new sequence analysis methods based on MCSs.  
\end{abstract}

\section{Introduction}
\label{sec:introduction}

The \emph{Longest Common Subsequence} (LCS) \cite{stringtostring, survey, hirschberg, fasterlcs} have thoroughly been studied in a plethora of string comparison application domains, like spelling error correction, molecular biology, and plagiarism detection, to name a few. The LCS is a special case of \emph{Maximal Common Subsequence} (MCS) for any two strings $X,Y$: it is a string $S$ that is a subsequence of both $X$ and $Y$, and is inclusion-maximal, namely, no other string $S'$ containing $S$ is also a common subsequence of $X$ and $Y$. The set of all MCSs for $X$ and $Y$ is denoted by $\MCS(X,Y)$. For example, $\MCS(X,Y) = \{\T\A\C\A,\G\}$ for $X = \T\C\A\C\A\G$ and $Y = \G\T\A\C\T\A$, whereas $\T\C\A \not \in \MCS(X,Y)$ as it is contained in $\T\A\C\A$, and the latter is also the only LCS. A crucial observation is the following: if $S$ is a common subsequence, it is always contained in some MCS but this is not necessarily true for LCS (e.g.~$S=\G$ is not contained in $\T\A\C\A$).

In general, LCSs only provide us with information about the longest possible alignment between two strings, while MCSs offer a different range of information, possibly revealing slightly smaller but alternative alignments. For very long strings an MCS may be just slightly shorter than an LCS but provide information on parts of the strings not found by LCSs. In principle, MCS could provide helpful information in all the applications where LCS are used.

While there is a quadratic conditional lower bound for the computation of LCS, based on the Strong Exponential Time Hypothesis~\cite{hardness}, no such bound exists for MCS: actually, an MCS between two strings of length $n$ can be extracted in $O(n\sqrt{\log n /\log\log n})$ time~\cite{SAKAI2019132}.
Moreover it is NP-hard to compute the LCS of an arbitrary number of strings~\cite{maier1978complexity}, whereas a recent polynomial-time algorithm for extracting an MCS from multiple strings exists~\cite{hirota2023fast}.
It is worth noting that even though there are a few more different approaches in the literature to find LCS or common subsequences with some kind of constraints (e.g. common subsequence trees \cite{hsu1984computing}, common subsequence automata \cite{crochemore2003directed}), the above observations motivate further investigation to directly deal with MCS. 

In this paper we proceed along that direction, and consider the problem of storing and searching the set $\MCS(X,Y)$. The main hurdle is that $\MCS(X,Y)$ could contain an exponential number of distinct strings~\cite{conte2022enumeration}: consequently, any trie-based or immediate automaton representation
would require \textit{exponential} time and space for its construction. We improve significantly over this direction as we list in our contributions, where $n = \max\{|X|,|Y|\}$ and $\sigma$ is the size of the alphabet $\Sigma$ for $X$ and $Y$:
\begin{itemize}
        \item We introduce a labeled compact direct acyclic graph $\MCSDAG(X,Y)$ (for $MCS$ $DAG$) that represents the strings in $\MCS(X,Y)$ in polynomial space $O(n^3 \sigma)$ and can be built in polynomial time $O(n^3 \sigma \log n)$. 
    \item For any string $P$, we show how $\MCSDAG(X,Y)$ can \textit{search} the strings with prefix $P$ from $\MCS(X,Y)$ and report them in lexicographic order, in $O(|P| \log \sigma + \mathit{occ})$ time, where $\mathit{occ}$ is the number of reported strings.
    \item For any integer $1 \leq i \leq |\MCS(X,Y|$, we show how $\MCSDAG(X,Y)$ with input $i$ can \textit{select} the $i$th string $S$ in lexicographic order from $\MCS(X,Y)$, in $O(|S| \log \sigma)$ time; the inverse operation of \textit{rank} for input $S$ is supported in the same complexity, where $i$ is returned.
    \item We can list all the strings from $\MCS(X,Y)$ lexicographically in constant amortized time (CAT), namely, $O(|\MCS(X,Y)|)$ time.
\end{itemize}

In Section~\ref{section:MCSDAG-construction} we implement $\MCSDAG(X,Y)$ as a direct acyclic graph (DAG) where the only zero indegree node is the source $s$ and the only zero outdegree node is the target $t$. Each edge is labeled with a character from the alphabet, so that any two outgoing edges from the same node have different characters as labels. Moreover, each $st$-path corresponds to a string in $\MCS(X,Y)$, obtained by concatenating the characters on the edge traversed along the $st$-path; vice versa, each string in $\MCS(X,Y)$ is spelled out by an $st$-path. In order to define and build $\MCSDAG(X,Y)$ we use some properties on MCSs previously introduced in~\cite{conte2022enumeration} and described in Section~\ref{section:cat-alg}, plus some new ideas to keep the size of $\MCSDAG$ polynomial. As a result, after $\MCSDAG$ has been built and its unary paths have been compacted, the aforementioned operations can be implemented in a simple way, as shown in Section~\ref{sec:operations}.

Finally, we observe that labels in $\MCSDAG(X,Y)$ may not have constant size each, so of course printing each reported string may require more than constant time. However, our algorithms are able to report a \textit{compressed} form of the output by simply printing the changes, similarly to~\cite{TOMITA200628}. On the other hand, we recall that the complexity of the algorithms should not measure the printing procedures, as formalized by Frank Ruskey~\cite[Section 1.7]{ruskey2003combinatorial} in the ``Don’t Count the Output Principle'', but rather the amount of data change that takes place. As the latter takes constant time per solution, we have CAT complexity.

\medskip

\noindent\textbf{Related work.} \quad
Maximal common subsequences were first introduced in~\cite{first}, in the context of LCS approximation. 
The first algorithm for finding an MCS between two strings was presented in~\cite{Sakai}, and subsequently refined in~\cite{SAKAI2019132}. The latter algorithm finds an MCS between two strings of length $n$ in $O(n\sqrt{\log n /\log\log n})$.
These algorithms can be also used to extend a given sequence to a maximal one in the same time, and to check whether a given subsequence is maximal in $O(n)$ time. 
In~\cite{hirota2023fast} the authors considered the problem of finding an MCS of $m>2$ strings of total length $n$, and were able to solve this in $O(mn\log n)$ time and $O(n)$ space.
As for MCS enumeration,~\cite{conte2022enumeration} showed that this task can be performed in $O(n\log n)$ delay and quadratic space.

The automaton approach has been used in literature to deal with subsequence-related problems. The Directed Acyclic Subsequence Graph (DASG) introduced in~\cite{baeza1991searching} is an automaton which accepts all subsequences of a given string $S$. 
Given a set of strings, it can also be generalized to accepting subsequences of any string in the set. 
Later, the common subsequence automata (CSA) was introduced~\cite{crochemore1999directed,crochemore2003directed, tronicek2002common}, which instead accepts \emph{common} subsequences of a set of strings. Such automaton is similar to the common subsequence tree of~\cite{hsu1984computing}, and it can also be used to find a longest common subsequence between two strings~\cite{melichar2003longest}. 
Binary decision diagrams such as ZDD~\cite{Minato93} and SeqBDD~\cite{LoekitoBP10} could potentially be employed to compactly store $\MCS(X,Y)$ but their worst-case behaviour typically involves exponential construction time and space.

We stress that it is not straightforward to efficiently adapt all of these structures to the MCS problem: indeed, figuring out which of the accepted strings are subsequences of each other is not an immediate task. We therefore opted for a different approach, directly defining and constructing an MCS automaton, based on combinatorial properties specific to MCSs.

\noindent\textbf{Preliminaries. }
[\textit{String notation}] 
A string $S$ over an alphabet $\Sigma$ (of size $\sigma=|\Sigma|$) is a concatenation of any number of its characters. The empty string is denoted with $\varepsilon$. For $i \in [0,|S|-1]$, character $S[i]$ occurs at position $i$ of string $S$,
and position $\nextc_S(c, i)$ denote the next occurrence of character $c$ after position $i$ in string $S$, if it exists (otherwise,  $\nextc_S(c,i) = |S|-1$). 
The notation $S_{<i}$ indicates $S[0, i-1]$, and $S_{\leq i}$ indicates $S[0,i]$.
    We say that $S$ is a \emph{subsequence} of a string $X$, denoted $S\subset X$, if there exist indices $0\leq i_0 < ... < i_{|S|-1} < |X|$ such that $X[i_k] = S[k]$ for all $k \in [0,|S|-1]$. In this case we also say that $X$ \emph{contains} $S$. Given two strings $X$ and $Y$, a string $S$ is a common subsequence if $S\subset X$ and $S\subset Y$. Furthermore, $S$ is a \emph{Maximal Common Subsequence}, denoted $S \in \MCS(X,Y)$, if it is a common subsequence which is \emph{inclusion-maximal}: there is no string $T \neq S$ such that $S \subset T \subset X,Y$.     
\quad [\textit{Graph notation}]
Given a directed graph $G=(V,E)$, we define the in-neighbors of a node $v\in V$  as $N^-(v) = \{u\in V \ | \ (u,v)\in E\}$, and the out-neighbors as $N^+(v) = \{v\in V \ | \ (v,z)\in E\}$; the in-degree is $d^-(v) = |N^-(v)|$ and the out-degree is $d^+(v) = |N^+(v)|$.
A subgraph of graph $G= (V,E)$ is a graph $H = (W,F)$ such that $W \subseteq V$ and $F \subseteq E$. 
A \emph{path} $P$ in a graph $G= (V,E)$ is a sequence of distinct adjacent nodes: $P = u_1 ... u_k$ where $u_i \in V$ for all $i$, $u_i \neq u_j$ for all $i \neq j$, and $(u_i, u_{i+1})\in E $ for all $i$. $P$ is called an $u_1u_k$-path.  A path where the first and last nodes coincide is called a \emph{cycle}. A directed graph with no cycles is called a directed acyclic graph, or \emph{DAG}. 

\section{Structure of MCSs}
\label{section:cat-alg}

In this section, we show known properties of MCSs which will be at the base of \MCSDAG.
Firstly, we show how to naturally represent all the strings in $\MCS(X,Y)$ for any two given strings $X$ and $Y$, using a finite-state automaton $A$ corresponding to a DAG with a single source $s$ and a single target $t$, where each edge is labeled with a character of the alphabet $\Sigma$; we show that $st$-paths in $A$ have a one-to-one correspondence with the strings in $\MCS(X,Y)$. 
Secondly, we recall and contextualize some useful properties of MCSs proven in~\cite{conte2022enumeration}.


\subsection{Modeling MCS as an Exponentially Large DAG}
\label{section:automaton-representation}

We here define a deterministic finite-state automaton $A$ that accepts the strings in $\MCS(X,Y)$. It is more convenient to define it directly as an equivalent DAG, using the extended alphabet $\Sigma' = \Sigma \cup \$$, where the dollar sign is a new special character to identify the end of an MCS. 
\begin{definition}
\label{definition:automaton}
The edge-labeled DAG $A = (V\cup \{s,t\},E,\edgelab{\cdot})$
is defined for two input strings $X$ and $Y$ to store their set $\MCS(X,Y)$, as follows:
\begin{itemize}
\item Every node in $V\cup \{s\}$ corresponds to a distinct prefix of a string in $\MCS(X,Y)$, with $s$ called \emph{source node} corresponding to the empty string prefix.
\item For any two nodes $u, u' \in V$, let $P, P'$ be their corresponding two prefixes. Then, $P' = P c$ iff $(u,u') \in E$ with label $\edgelab{u,u'} = c$. 
\item For any node $u\in V$, if its corresponding prefix is a whole string $P \in \MCS(X,Y)$, then $(u,t) \in E$ with label $\edgelab{m,t} = \$$. 
\end{itemize}
\end{definition}


\begin{figure}
\centering
\begin{tikzpicture}[xscale = 1.5, scale = 0.55]
        \vertex[scale = 2] (s) at (0,0.5) {\tiny $s$};
        \vertex[scale = 2] (1) at (1,2.7) {};
        \vertex[scale = 2] (2) at (1,0.5) {};
        \vertex[scale = 2] (4) at (2,3.5) {}; 
        \vertex[scale = 2] (5) at (2,0.5) {};
        \vertex[scale = 2] (7) at (3,4) {};
        \vertex[scale = 2] (21) at (3,1) {}; 
        \vertex[scale = 2] (8) at (3,0) {}; 
        \vertex[scale = 2] (10) at (4,4) {};
        \vertex[scale = 2] (18) at (3,3) {}; 
        \vertex[scale = 2] (11) at (4,3) {}; 
        \vertex[scale = 2] (12) at (4,1) {};
        \vertex[scale = 2] (13) at (4,0) {};
        \vertex[scale = 2] (16) at (5,4) {};
        \vertex[scale = 2] (17) at (5,3) {};
        \vertex[scale = 2] (19) at (5,1) {};
        \vertex[scale = 2] (20) at (5,0) {};
        \vertex[scale = 2] (6) at (2,-2) {}; 
        \vertex[scale = 2] (9) at (3,-2) {};
        \vertex[scale = 2] (15) at (4,-2) {};
        \vertex[scale = 2] (23) at (5,-2) {};
        
        \vertex[scale = 2] (t) at (7,0.5) {\tiny $t$};

        \path[->] 
        (s) edge[bend left = 10] node [left] {\texttt{A}} (1)
        (s) edge node [above] {\texttt{C}} (2)
        (1) edge[bend left = 8] node [above] {\texttt{C}} (4)
        (2) edge node [above] {\texttt{C}} (5)
        (s) edge[bend right = 15] node [below] {\texttt{T}} (6)
        (4) edge node [above] {\texttt{A}} (7)
        (4) edge node [below] {\texttt{G}} (18)
        (5) edge node [above] {\texttt{A}} (21) 
        (5) edge node [below] {\texttt{G}} (8) 
        (21) edge node [above] {\texttt{G}} (12)
        (6) edge node [above] {\texttt{A}} (9)
        (7) edge node [above] {\texttt{G}} (10)
        (8) edge node [below] {\texttt{A}} (13)
        (9) edge node [above] {\texttt{G}} (15)
        (10) edge node [above] {\texttt{G}} (16)
        (18) edge node [below] {\texttt{A}} (11)
        (11) edge node [below] {\texttt{G}} (17)
        (12) edge node [above] {\texttt{G}} (19)
        (13) edge node [below] {\texttt{G}} (20)
        (15) edge node [above] {\texttt{G}} (23)
        (16) edge[bend left = 25] node [above] {\texttt{\$}} (t) 
        (17) edge[bend left = 10] node [above] {\texttt{\$}} (t) 
        (19) edge[bend left = 10] node [above] {\texttt{\$}} (t) 
        (20) edge[bend right = 10] node [above] {\texttt{\$}} (t) 
        (23) edge[bend right = 15] node [above] {\texttt{\$}} (t) 
        ;
    \end{tikzpicture}
    \hfill
\begin{tikzpicture}[xscale = 1.5, scale = 0.65]
    \vertex[scale = 2] (s) at (0,0) {\tiny $s$};
    \vertex[scale = 2] (1) at (2,0) {};
    \vertex[scale = 2] (2) at (3.5,-1) {};
    \vertex[scale = 2] (t) at (5,0) {\tiny $t$};
    \node at (0,-4) {};

    \path[->] 
    (s) edge[bend left = 60] node [above] {\texttt{AC}} (1)
    (s) edge[bend right = 60] node [above]  {\texttt{CC}} (1)
    (s) edge[bend right = 60] node [below] {\texttt{TAG}} (2)
    (1) edge[bend right = 30] node [above] {\texttt{AG}} (2)
    (1) edge[bend left = 50] node [above] {\texttt{GAG\$}} (t) 
    (2) edge[bend right = 30] node [above] {\texttt{G\$}} (t)
    ;
\end{tikzpicture}
\caption{Left: edge-labeled DAG $A$ for strings $X=$ \texttt{TCACAGAGA} and $Y=$ \texttt{ACCCGTAGG}. We have $\MCS(X,Y) =  \{\texttt{ACAGG}, \texttt{ACGAG}, \texttt{CCAGG}, \texttt{CCGAG}, \texttt{TAGG}\}$; each $st$-path in $A$ corresponds to one of these. For instance, the topmost path corresponds to \texttt{ACAGG}. 
Right: compact $\MCSDAG$ for the same strings.}
\label{fig:MCSDAG}
\end{figure}
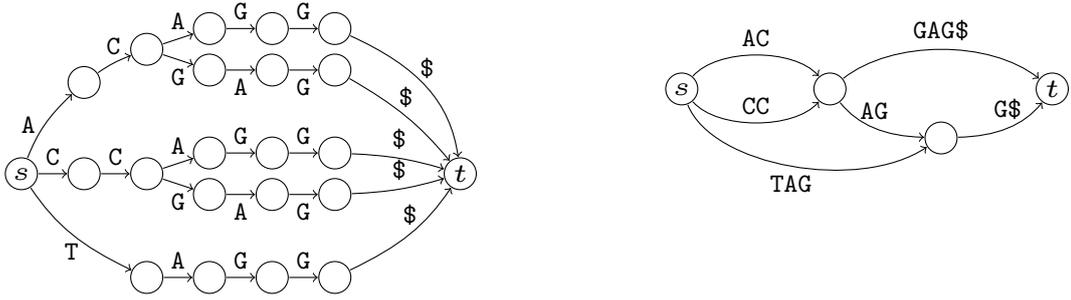

Given DAG $A$, we have a bijection between the strings in $\MCS(X,Y)$ and the labeled $st$-paths (see the left of Figure~\ref{fig:MCSDAG} for an example).
Indeed, it is immediate that any such path corresponds to a string in $\MCS(X,Y)$ due to the way edges are placed. Vice versa, by definition, we cannot have two out-going edges from the same node that share the same label. Therefore, each MCS has exactly one corresponding $st$-path. 
Note that the property of the outgoing labels also implies that $d^+(u) \le  \sigma$ for all $u \in V$, and therefore $|E| \le \sigma |V|$. As the number of nodes in $A$ is $\Omega(|\MCS(X,Y)|)$, its size is exponential in the worst case. 

\begin{remark}
$A$ can be seen as an acyclic deterministic finite-state automaton that accepts the set $\MCS(X,Y)$. We have one state for each prefix of a string in $\MCS(X,Y)$. Given two states $q, q'$, respectively corresponding to prefixes $P, P'$ of strings in $\MCS(X,Y)$, 
the transition function $\delta$ is given by $\delta(q, c) = q'$ if and only if $P' = P c$.
\end{remark}



\subsection{Known Concepts on MCS Enumeration}
\label{sec:previous-algo}
We now present a summary of some concepts introduced in~\cite{conte2022enumeration}, which will be used for our results. 
Let string $P$ be called a \emph{valid prefix} if there exists a string $Q$ such that their concatenation is $PQ\in \MCS(X,Y)$. Given a valid prefix $P$, the set of characters $c$ such that $P \, c$ is still a valid prefix are called \emph{valid extensions}. 

Finding valid extensions of $P$ is not straightforward.
Let $P$ be a valid prefix, and let $X_{\le l}$ and $Y_{\le m}$ be respectively the shortest prefixes of $X$ and $Y$ that contain $P$ as a subsequence. Consider a pair of positions $i > l$ and $j > m$, respectively in $X$ and $Y$, corresponding to the same character, say,~$c \in \Sigma$. In order to assess whether $c$ is a valid extension, checking if $P \, c$ is a maximal common subsequence of $X_{\le i}$ and $Y_{\le j}$ is necessary, but not sufficient:

\begin{example}
\label{exa:pitfall}
Consider $X = \texttt{TCACAG}$ and $Y = \texttt{TACGAT}$, with $\MCS(X,Y ) = \{\texttt{TACA}, \texttt{TACG}\}$.

\noindent\begin{minipage}{.3\textwidth}
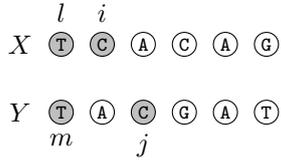
    \begin{tikzpicture}[square/.style={regular polygon,regular polygon sides=4},x=.6cm, y=1cm, scale = 0.9]
\node at (-1,1) {$X$};
    \prefixvertex[scale = 1.5] (xbeg) at (0,1) [label = above:  $l$] {\tiny \texttt{T}};
    \prefixvertex[scale = 1.5] (x1) at (1,1) [label = above:  $i$] {\tiny \texttt{C}};
    \vertex[scale = 1.5] (x2) at (2,1) {\tiny \texttt{A}};
    \vertex[scale = 1.5] (x3) at (3,1) {\tiny \texttt{C}};
    \vertex[scale = 1.5] (x4) at (4,1) {\tiny \texttt{A}};
    \vertex[scale = 1.5] (x5) at (5,1)   {\tiny \texttt{G}};
    
\node at (-1,0) {$Y$};  
    \prefixvertex[scale = 1.5] (ybeg) at (0,0) [label = below:  $m$] {\tiny \texttt{T}};
    \vertex[scale = 1.5] (y1) at (1,0) {\tiny \texttt{A}};
    \prefixvertex[scale = 1.5] (y2) at (2,0) [label = below:  $j$] {\tiny \texttt{C}};
    \vertex[scale = 1.5] (y3) at (3,0)  {\tiny \texttt{G}};
    \vertex[scale = 1.5] (y4) at (4,0) {\tiny \texttt{A}};
    \vertex[scale = 1.5] (y5) at (5,0)  {\tiny \texttt{T}};
   
    \path
        %
        ;
\end{tikzpicture}
\end{minipage}
\begin{minipage}{.66\textwidth}
     Consider the valid prefix $\T$, for which $l=0$ and $m=0$. Let $i=1$ and $j = 2$: these correspond to character $\C$, and clearly $\T\C \in \MCS(X_{\le i}, Y_{\le j}) = \MCS(\texttt{TC}, \texttt{TAC})$. Still, $\C$ is not a valid extension of valid prefix $\T$, as $\T\C$ is not a prefix of any MCS. 
\end{minipage}
\end{example}

To circumvent this problem, the set $\Ext_{l,m}$ of \emph{candidate extensions} is defined in~\cite{conte2022enumeration}: this is a set of pairs of positions $(i,j)$, with $i \ge l$ and $j \ge m$, whose definition relies solely on the pair $(l,m)$. The membership of $(i,j)$ to $\Ext_{l,m}$ completes the characterization of valid extensions (see Theorem 3 in~\cite{conte2022enumeration}): 
\begin{itemize}
    \item Let $P$ be a valid prefix, and let $X_{\le l}$ and $Y_{\le m}$ be respectively the shortest prefixes of $X$ and $Y$ that contain $P$ as a subsequence.
    \item Then $P\, c$ is a valid prefix if and only if the following two conditions hold:
    \begin{enumerate}
        \item \label{condition:swings} There exists $i$ and $j$ such that $c = X[i] = Y[j]$ and  $P \in \MCS(X_{<i}, Y_{<j})$;
        \item \label{condition:ext} This pair of positions satisfies $(i,j) \in \Ext_{l,m}$.
    \end{enumerate} 
\end{itemize}
More details on the construction of set $\Ext_{l,m}$ are beyond the scope of this paper, but they can be found in~\cite[Section 2.3]{conte2022enumeration}, along with a  $O(\sigma \log n)$ time method for its computation. In Example~\ref{exa:pitfall}, it is $(i,j) \not \in \Ext_{l,m}$.

\begin{figure}
    \centering
\begin{tikzpicture}[square/.style={regular polygon,regular polygon sides=4},x=.6cm, y=1cm, scale = 0.9]
\node at (-1,1) {$X$};
    \prefixvertex[scale = 1.5] (xbeg) at (0,1) {\tiny \texttt{T}};
    \prefixvertex[scale = 1.5] (x1) at (1,1) {\tiny \texttt{A}};
    \vertex[scale = 1.5] (x2) at (2,1) {\tiny \texttt{T}};
    \prefixvertex[scale = 1.5] (x3) at (3,1) [label = above: $l$] {\tiny \texttt{C}};
    \vertex[scale = 1.5] (x4) at (4,1) {\tiny \texttt{G}};
    \vertex[scale = 1.5] (x5) at (5,1)   {\tiny \texttt{A}};
    
    \vertex[scale = 1.5] (x6) at (6,1)  [label = above: \small $\ltimes_T$] {\tiny \texttt{C}};
    \vertex[scale = 1.5] (x7) at (7,1) {\tiny \texttt{T}};
    \vertex[scale = 1.5] (x8) at (8,1)  {\tiny \texttt{C}};
    
\node at (-1,0) {$Y$};  
    \prefixvertex[scale = 1.5] (ybeg) at (0,0) {\tiny \texttt{T}};
    \vertex[scale = 1.5] (y1) at (1,0) {\tiny \texttt{G}};
    \prefixvertex[scale = 1.5] (y2) at (2,0) {\tiny \texttt{A}};
    \prefixvertex[scale = 1.5] (y3) at (3,0) [label = below:  $m$] {\tiny \texttt{C}};
    \vertex[scale = 1.5] (y4) at (4,0) {\tiny \texttt{G}};
    \vertex[scale = 1.5] (y5) at (5,0)  {\tiny \texttt{C}};
    \vertex[scale = 1.5] (y6) at (6,0)  {\tiny \texttt{T}};
    \vertex[scale = 1.5] (y7) at (7,0) {\tiny \texttt{A}};
    \vertex[scale = 1.5] (y8) at (8,0)  [label = below:  \small $\ltimes_B$]  {\tiny \texttt{C}};
    
    \path
        (x6) edge[color = black, line width = 1pt] (y3)
        (x5) edge[dashed, color = black, line width = 1pt] (y2)
        (x4) edge[dashed, color = black, line width = 1pt] (y1)
        ;
    
                
\end{tikzpicture}
\hspace{1cm}
\begin{tikzpicture}[square/.style={regular polygon,regular polygon sides=4},x=.6cm, y=1cm, scale = 0.9]
    \vertex[scale = 1.5] (xbeg) at (0,1) {\tiny \texttt{T}};
    \vertex[scale = 1.5] (x1) at (1,1) {\tiny \texttt{C}};
    \vertex[scale = 1.5] (x2) at (2,1) {\tiny \texttt{A}};
    \vertex[scale = 1.5] (x3) at (3,1) {\tiny \texttt{C}};
    \vertex[scale = 1.5] (x4) at (4,1) {\tiny \texttt{A}};
    \vertex[scale = 1.5] (x5) at (5,1)  [label = above: \color{red} $l$] {\tiny \texttt{G}};
    
    \vertex[scale = 1.5] (x6) at (6,1)   {\tiny \texttt{A}};
    \vertex[scale = 1.5] (x7) at (7,1) {\tiny \texttt{T}};
    \vertex[scale = 1.5] (x8) at (8,1) [label = above: \color{red}  \small $\ltimes_T$] {\tiny \texttt{G}};
    
    
    \vertex[scale = 1.5] (ybeg) at (0,0) {\tiny \texttt{A}};
    \vertex[scale = 1.5] (y1) at (1,0) {\tiny \texttt{C}};
    \vertex[scale = 1.5] (y2) at (2,0) {\tiny \texttt{T}};
    \vertex[scale = 1.5] (y3) at (3,0) {\tiny \texttt{C}};
    \vertex[scale = 1.5] (y4) at (4,0) {\tiny \texttt{T}};
    \vertex[scale = 1.5] (y5) at (5,0) [label = below: \color{red} $m$] {\tiny \texttt{G}};
    \vertex[scale = 1.5] (y6) at (6,0)  {\tiny \texttt{G}};
    \vertex[scale = 1.5] (y7) at (7,0) {\tiny \texttt{T}};
    \vertex[scale = 1.5] (y8) at (8,0) {\tiny \texttt{A}};
    \vertex[scale = 1.5] (y9) at (9,0) [label = below: \color{red} \small $\ltimes_B$] {\tiny \texttt{G}};
    \path
        (xbeg) edge[color = blue, line width = 1pt] (y2)
        (x1) edge[color = blue, line width = 1pt] (y3)
        (x2) edge[dashed, color = orange, line width = 1pt] (ybeg)
        (x3) edge[dashed, color = orange, line width = 1pt] (y1)
        (x5) edge[color = red, line width = 1.5] (y5)
        (x7) edge[dotted, color = green!70!black, line width = 0.8pt] (y7)
        (x6) edge[dotted, color = green!70!black, line width = 0.8pt] (y8);
    
                
\end{tikzpicture}
    \caption{
    Left: Swings $(\ltimes_T, \ltimes_B)$ for valid prefix $\texttt{TAC}$ in strings $X=\texttt{TATCGACTC}$ and $Y = \texttt{TGACGCTAC}$. Consider swing $\ltimes_T$: $\texttt{TAC}\not\in \MCS(X_{\le \ltimes_T}$, $Y_{\le m})$ since  $\texttt{TGAC}$ is a common subsequence (dashed).
    Right: Two prefixes, \texttt{TCG} (solid blue) and \texttt{ACG} (dashed orange), both ending at solid red positions $(l,m)$, and having the same swings $(\ltimes_T, \ltimes_B)$. The valid extensions are the same (dotted green): \texttt{A} and \texttt{T}.}
    \label{fig:swings}
\end{figure}
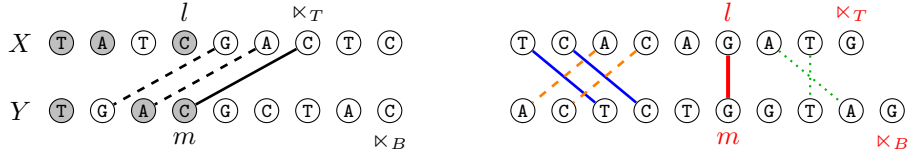

The notion of \emph{swings} is a key concept to quickly verify the first condition above. Swings characterize the amount by which we can ``move'' a given character occurrence while retaining maximality (see Figure~\ref{fig:swings}, right). 
Let $X_{\le l}$ and $Y_{\le m}$ be respectively the shortest prefixes of $X$ and $Y$ that contain $P$ as a subsequence. 
The \emph{swing} of $P$, denoted $\ltimes(P)$, is a pair of integers  $\ltimes_T(P)$ and $\ltimes_B(P)$, called respectively \emph{top} and \emph{bottom} swings, given by
\begin{itemize}
    \item $\ltimes_T(P) = \min \{ i> l \ | \ P \not \in \MCS(X_{\le i},Y_{\le m})\}$ 
    \item $\ltimes_B(P) = \min \{j>m \ | \ P \not \in \MCS(X_{\le l},Y_{\le j})\}$ 
\end{itemize}
It follows from the definition that $P \in \MCS(X_{< i}, Y_{< j}) \iff i \le \ltimes_T (P) \text{ and } j \le  \ltimes_B(P)$. Given the swings, condition~\ref{condition:swings} can thus be checked in constant time~\cite{conte2022enumeration}. 

Updating the swings when a character is added can be done in $O(\sigma)$ time, provided that we can find in constant time the next occurrence of a character $c$ after a given position in the strings; more importantly, this update does \textit{not} need knowledge of the whole prefix, but just the positions its final character ($l_N,m_N$) and their current swings. 





\section{Polynomial-Size \MCSDAG}
\label{section:MCSDAG-construction}
The construction of DAG $A$ satisfying the conditions of Definition~\ref{definition:automaton} would require exponential time and space: the number of nodes of $A$ is between  $\Omega(|\MCS(X,Y)|)$ and $O(n|\MCS(X,Y)|)$, so it can be exponential in $n$. In this section, we show how to obtain $\MCSDAG$, where we still have a bijection between $st$-paths and MCS, but which can instead always be constructed in $O(n^3 \sigma \log n )$ time and $O(n^3\sigma)$ space, as per Theorem~\ref{thm:main-CAT}. Intuitively, the relevant information discussed for $A$
are the quadruples $(l,m, t, b)$, where $X_{\le l}$ and $Y_{\le m}$ are some prefixes and pair $t, b$ is some swing: these quadruples are the candidates for being nodes in $\MCSDAG$.

The formal definition of \MCSDAG is based on an equivalence relation over the nodes of $A$, given in Section~\ref{section:mcsdag-eqrel}. 
Afterwards, in Section~\ref{section:mcsdag-construction}, we describe an algorithm for \textit{directly} constructing $\MCSDAG$. We present the complexity bounds for the construction in Section~\ref{section:mcsdag-size}. 

\subsection{Equivalence Relation for Defining \MCSDAG}
\label{section:mcsdag-eqrel}
Let $A$ be a DAG as defined in Definition~\ref{definition:automaton}.
Our construction algorithm for $\MCSDAG$ is based on the concepts from Section~\ref{sec:previous-algo}. 
The idea is to use the characterization of valid extensions to identify the out-neighbors of a given node of DAG $A$. 
%
%
%
We identify an equivalence relation over the prefixes of $\MCS(X,Y)$, and thus on the nodes of $A$, that allow us to always build $\MCSDAG$ in polynomial time and space.  
We begin with the following lemma:
\begin{lemma}
\label{lemma:equal-extensions}
    Given any valid prefix $P$, let $X_{\le l}$ and $Y_{\le m}$ be the shortest prefixes containing $P$, and
    $\ltimes(P) = \langle t, b \rangle$ be its swing. Consider another valid prefix $P'\neq P$ with the \emph{same} shortest prefixes $X_{\le l}, Y_{\le m}$ and swing $\ltimes(P') = \ltimes(P)$ as $P$. Then, the set of valid extensions is the same for both $P$ and $P'$, and for each valid extension $c \in \Sigma$, the swings of $P\, c$ are the same as the ones of $P'\, c$. 
\end{lemma}
\begin{proof}
The definition of $\Ext_{l,m}$ only depends on the value of $l$ and $m$, therefore such set is the same for both $P$ and $P'$. Since the swings form $P$ and $P'$ are equal, the set of valid extensions is necessarily the same. 
Let now $c \in \Sigma$ be a valid extension for $P$ and $P'$. Let $X_{\le l_c}$ and $Y_{\le m_c}$ respectively be the shortest prefixes of $X$ and $Y$ containing $P\, c$. These are also the shortest prefixes containing $P' \, c$: the shortest prefixes containing $P$ and $P'$ were the same, and $l_c$ is simply the first occurrence of $c$ after $l$, analogously for $m_c$ and $m$. 
The swings of $P \, c$ are given by the minimum of the swings of $P$, and the personal swing $\ltimes{(l_N, m_N)}$ obtained by adding the new character $c$. The latter personal swing is the same for both $P \, c$ and $P'\, c$, since we are considering the same positions $l_c,m_c$. Since the previous swings where also equal, this means that the swings of $P \,c $ and $P' \, c$ are indeed the same. 
\end{proof}

Lemma~\ref{lemma:equal-extensions} has an implication for prefixes $P$ and $P'$ that share the \textit{same} swing: if $M_1,...,M_N$ are strings extending as $PM_i \in \MCS(X,Y)$, and $M'_1,...,M'_M$ extending as $P'M'_i \in \MCS(X,Y)$, then they are equal: $\{M_i \ | \ i = 1,...,N\}=\{M'_i \ | \ i = 1,...,M \}$.

Given a node $u$ of $A$, let $P$ be the corresponding prefix.  We assign $u$ the quadruple of parameters $\nodelab{u} = \langle l,m,\tswing{P}, \bswing{P}\rangle$, where $l$ and $m$ are such that $X_{\le l}$ and $Y_{\le m}$ are the shortest prefixes containing $P$. 
By Lemma~\ref{lemma:equal-extensions}, this tuple completely identifies the valid extensions of $P$, which means that it completely identifies the neighbors of node $u$. 
\begin{corollary}
\label{corollary:same-neighbors}
    Let $u \neq u'$ with $\nodelab{u} = \nodelab{u'}$. Then, for each $v \in N^+(u)$ there exists exactly one $v' \in N^+(u')$ such that $\nodelab{v} = \nodelab{v'}$ and $\edgelab{u,v} = \edgelab{u',v'}$.
\end{corollary} 

Therefore, we can define the following \emph{equivalence relation} on the nodes of $A$: $u \sim u'$
if and only if $\nodelab{u} = \nodelab{u'}$. We can then identify a class of equivalent nodes in the DAG, choosing one representative for it. Because of Corollary~\ref{corollary:same-neighbors}, this does not change the set of labeled $st$-paths of the DAG: the nodes that are identified as one have the same labelled out-edges, leading to the same out-neighbors. Our data structure \MCSDAG is then defined as the DAG resulting from this identification:
\begin{definition}
\label{definition-MCSDAG}
Data structure $\MCSDAG$ is a node- and edge-labelled DAG  built as follows:
\begin{enumerate}
\item Start from DAG $A$ (Definition~\ref{definition:automaton}). For each node $u$, consider its (unique) corresponding prefix $P$, and let $X_{\le l}$ and $Y_{\le m}$ be the shortest prefixes of $X$ and $Y$ containing $P$, and $\ltimes(P) = (t,b)$. Assign to node $u$ the node-label $\nodelab{u}=\recID{l}{m}{t}{b}$. \label{item:MCSDAG-start}
\item Merge every pair of nodes $u \neq u'$ with the same label $\nodelab{u} = \nodelab{u'}$ into one node.  \label{item:MCSDAG-identify}
\end{enumerate}
An example of such DAG is shown in the right of Figure~\ref{fig:MCSDAG}.
\end{definition}

The \textsf{compact} \MCSDAG is the version with compressed \emph{unary paths}. Each such path $u_1...u_k$ has $N^+(u_i) = \{u_{i+1}\}$ for each $i$, and is replaced by the single edge $(u_1,u_k)$, with label $\edgelab{u_1,u_k} = \edgelab{u_1,u_2}\edgelab{u_2,u_3}...\edgelab{u_{k-1},u_k}$.
Note that compressing unary paths does not change the set of labeled $st$-paths of the DAG, which still correspond to $\MCS(X,Y)$.

\subsection{Direct and Incremental Construction of \MCSDAG}
\label{section:mcsdag-construction}
We build $\MCSDAG$ directly, without the intermediate DAG $A$, We apply the incremental procedure below, in a DFS fashion, using the node $ID$s to avoid repeated computation.
At any moment we have built a node- and edge-labeled DAG $H = (W,F)$.
A node $u\in W$ corresponds to a set of prefixes $P_1,...,P_k$, given by the concatenation of the edge-labels of all $su$-paths using edges of $F$. All prefixes $P_i$ share the same ending positions of the shortest prefixes of $X$ and $Y$ that contain them, and the corresponding swings; these four values form the label $\nodelab{u}$ assigned to $u$.

Every recursive call \graphbuild{u} takes as input a node $u$ which belongs to the current DAG $H$, and expands DAG $H$ accordingly as follows:
\begin{enumerate}
    \item Let $\nodelab{u} = \recID{l}{m}{t}{b}$. First, compute set $\Ext_{l,m}$, and use it to compute the valid extensions: select characters $c$ that have a corresponding pair $(i,j) \in \Ext_{l,m}$, with $i \le t$ and $j \le b$.

    \item For such character $c$, compute the positions $(l_c,m_c)$ such that $X_{\le l_c}$ and $Y_{\le m_c}$ are the shortest prefixes containing $P\, c$, and update the swings $t_c, b_c$. 

    \item Now, check if the DAG $H$ generated so far already has a node with label $\recID{l_c}{m_c}{t_c}{b_c}$:
    \begin{enumerate}
        \item If such a node $v\in W$ exists, then simply add edge $(u,v)$ with label $c$ to the edges $F$ of $H$, without recursing. 
        Indeed, a recursive call for $v$ has been previously performed.  

        \item Otherwise, add node $v$ to $W$, with $\nodelab{v} = \recID{l_c}{m_c}{t_c}{b_c}$, and add edge $(u,v)$ to $F$, with label $c$. Then, perform the recursive call \graphbuild{v}. 
    \end{enumerate}

\end{enumerate}

\begin{corollary}[Correctness]
$\graphbuild{s}$ correctly builds $\MCSDAG(X,Y)$ starting from $H= (\{s\}, \emptyset)$. 
\end{corollary}
\begin{proof}
We show that, at every step, $H$ is a subgraph of DAG $A$ where nodes with equal $\nodelab{}$ values have been identified. This is true at the beginning, when $H =\{s\}$. If this holds at the beginning of a recursive call for a node $u$, then it must hold at the end. Indeed, we use valid extensions to identify neighbors of a given node, which is also the definition of neighbors for a node of $A$. For each valid extension $c$, we compute the label $\recID{l_c}{m_c}{t_c}{b_c}$ of the corresponding node. Now, if there already exists a node $v$ with such label, we add an edge between $u$ and $v$. This correctly identifies the two nodes with equal labels as being the same node, as per operation~\ref{item:MCSDAG-identify}. Otherwise, the new node must be  added, with the correct label $\recID{l_c}{m_c}{t_c}{b_c}$ as per operation~\ref{item:MCSDAG-start}, since it represents a new prefix. 
\end{proof}

Once we have built $\MCSDAG$, we can easily compute \textsf{compact} $\MCSDAG$ in linear time in the size of $\MCSDAG$, in the following way. Let us proceed in topological order of the nodes; when a node $v$ with $N^+(v) = \{w\}$ (i.e. out-degree 1) is encountered, we perform the following:
\begin{enumerate}
    \item For each $u \in N^-(v)$, remove edge $(u,v)$ and add edge $(u, w)$ with label $\edgelab{u,v}\edgelab{v,w}$.
    \item After all in-neighbors have been processed, remove node $v$ and edge $(v,w)$ from $\MCSDAG$. 
\end{enumerate}
Clearly, the minimum out-degree of  is 2,
and nodes $s$ and $t$ are never removed. 

\medskip 
\noindent\textbf{Complexity.} \quad Let us now study the time and space complexity of the procedure. As mentioned before, finding $\Ext_{l,m}$ requires $O(\sigma \log n)$ time (see \cite{conte2022enumeration} for details). Checking whether each $c$ corresponding to an element of $\Ext_{l,m}$ satisfies the swings' condition requires constant time per character. 
For each such $c$, computing positions $(l_c,m_c)$ can be done in constant time, while updating the swings accordingly can be performed in $O(\sigma)$ time, provided appropriate data structures. Indeed, to ensure this we only need to ensure constant-time queries for the next occurrence of character $c$ after a given position $i$. To this end, let us keep two bit-vectors for each $c \in \Sigma$, one for $X$ and one for $Y$, indicating the positions in which $c$ occurs in the strings. By equipping these vectors with rank and select data structures, which employ $O(n\sigma)$ space, we can find in constant time the next occurrence of any character after a given position~\cite{rankselect}.
All operations described so far are performed exactly once per node. Therefore, the total time required for these is $O(|V|\sigma \log n)$.

Let us now consider Step 3. We need to check if a node $v$ belongs to the current DAG, and add it if it does not. To be able to efficiently perform these operations, let us keep and dynamically update a bit matrix for pairs $l$, $m$, where 1 occurs if that pair currently corresponds to at least one node. Then, each cell filled with a one has an associated balanced binary search tree, which indexes the pairs of swings $(t,b)$ such that there currently is a node $u \in W$ with $\nodelab{u} =\recID{l}{m}{t}{b}$. These pairs are ordered according to the total order: $(t,b) < (t',b')$ if and only if $t < t'$ or $t=t'$ and $b < b'$.
Lookup and insertions in such data structure require $O(\log n)$ each, and the total space employed is $O(|V|)$. Now, note that we perform a membership check, with subsequent possible insertion, exactly once per edge of the DAG. Recalling that $|E|\le \sigma |V|$, the total time required by these operations is again $O(|E|\log n) = O(|V| \sigma \log n)$. 

To be able to give final complexity bounds, we therefore need bounds on the size of the DAG we constructed. Trivially, we can bound $|V| = O(n^4)$, since no two nodes share the same $ID$, and the number of different $ID$s is bounded by $n^4$. Therefore, we surely have a polynomial-time and space algorithm for building a DAG. 
We can actually do better than this: thanks to some properties of the swings, in the next section we show that $|V|=O(n^3)$, leading to the complexity bounds given in Theorem~\ref{thm:main-CAT}.

\subsection{Cubic Size of the MCS DAG}
\label{section:mcsdag-size}
To conclude the proof of Theorem~\ref{thm:main-CAT}, we need to study the size of $\MCSDAG(X,Y)$ as constructed in Section~\ref{section:mcsdag-construction}.
%
We prove a monotonicity property of the swing values, which will allow us to show that the number of nodes of $\MCSDAG(X,Y)$ is bounded by $O(n^3)$.

In Section~\ref{sec:previous-algo}, we saw that the top swing of a prefix $P$ is defined as $\ltimes_T(P) = \min \{ i> l \ | \ P \not \in \MCS(X_{\le i},Y_{\le m})\}$, where $X_{\le l}$ and $Y_{\le m}$ are the shortest prefixes of respectively $X$ and $Y$ containing $P$. In other words, if we start from strings $X_{\le l}, Y_{\le m}$ (where $P$ is obviously maximal), it is the minimum extension of string $X$ that ensures at least one insertion in $P$. 
Symmetrical definition holds for bottom swings, by switching the two strings. 

Note that, if $\ltimes_T(P) = t$, then $X[t] = Y[m] = X[l]$: the swings' positions are  occurrences of the last character of the prefix. Indeed, the swing is the first occurrence of character $Y[m] = X[l]$, after certain positions in string $X$ (see the incremental computation of swings in Section~\ref{sec:previous-algo}). We also note the following, which follows from the definition of swings:
\begin{remark}
\label{remark:remapping}
Let $P = p_1 \cdots p_N$ a valid prefix with swings $\langle t,b \rangle $. Let $X_{\le l_i}$ and $Y_{\le m_i}$ be the shortest prefixes respectively of $X$ and $Y$ that contain $p_1 \cdots p_i$. 
Then, there is at least one match between $Y[m_{N-1}, m_N)$ and $X(l_N, t)$. More specifically, there can either be a match between $Y(m_{N-1}, m_N)$ and $X(l_N, t)$, or between $Y[m_{N-1}]$ and $X(l,t)$ which will lead to an insertion in a previous part of the prefix.
\end{remark}


We briefly also recall how to incrementally compute the swings of a prefix, first described in~\cite{conte2022enumeration}, as it will be useful in our proofs. We only describe the procedure for the top swing as the bottom swing is symmetrical:
\begin{itemize}
\item If $P$ is composed of a single character $c$, with first occurrence in $X$ at position $l$ and in $Y$ at position $m$, then it suffices to compute, for every character $d$ in $Y_{<m}$, the first occurrence of $d$ in $X_{>l}$, and take the minimum of these. The swing then corresponds to the first occurrence of $c$ after such minimum.

\item Let $P= p_1 \cdots p_N$ be a valid prefix with $N>1$, and let $l_1,...,l_N$ (resp. $m_1,...,m_N$) be the positions of $X$ (resp. $Y$) such that $X_{\le l_i}$ (resp. $Y_{\le m_i}$) is the shortest prefix containing $p_1\cdots p_i$. The \emph{personal swing} $\ltimes_T{(l_N, m_N)}$ of the last position is the swing of character $p_N$ when seen as a prefix over the strings $X_{>l_{N-1}}, Y_{>m_{N-1}}$, instead of over the whole strings (and thus computed as above). In other words, the personal swing of a character expresses the change necessary to have an insertion between itself and the previous character of the prefix. The top swing of $P$ is  the minimum between the personal swing of $(l_N,m_N)$, and the first occurrence of $Y[m_N]=p_N$ after the top swing of prefix $p_1 \cdots p_{N-1}$. This second swing indicates the change required for an insertion in the previous part of the prefix. 
\end{itemize}


We now present some new swing properties. This lemma proves that, when two prefixes are extended with a valid extension that occurs at the same pair of positions, then the relative order of the swings remains unchanged during the extension:
\begin{lemma}
\label{lemma:monotone-swings}
Let $u$ and $u'$ be two nodes of \MCSDAG, with $\nodelab{u} = \recID{x}{y}{\tswing{P}}{\bswing{P}}$ and $\nodelab{u'} = \recID{x'}{y'}{\tswing{P'}}{\bswing{P'}}$ such that $x \neq x'$ or $y \neq y'$, where $P$ (resp. $P'$) is any prefix associated to $u$ (resp. $u'$). Assume that we have $v \in N^+(u)$ with $\nodelab{v} = \recID{l}{m}{\tswing{P\, c}}{\bswing{P\, c}}$, and $v' \in N^+(u')$ with $\nodelab{v'} = \recID{l}{m}{\tswing{P'\, c}}{\bswing{P'\, c}}$ (i.e.~same positions $l,m$ corresponding to character $c$). Then, the swings change monotonically:
\begin{itemize}
    \item $\tswing{P} < \tswing{P'} \Rightarrow \tswing{P \, c} \le \tswing{P' \, c}$
    \item $\bswing{P} < \bswing{P'} \Rightarrow \bswing{P \, c} \le \bswing{P' \, c}$
\end{itemize}
\end{lemma}
\begin{proof}
We prove the result for top swings. Let $\tswing{P} = t$ and $\tswing{P'} = t'$, with $t < t'$. We note that for $(l,m)$ to be a valid extension for both prefixes, we must have $l \le t < t'$ (Swing condition~\ref{condition:swings} for valid extensions' characterization in Section~\ref{sec:previous-algo}).
By the incremental computation of swings, the swings of $P \, c$ and $P' \, c$ are computed by taking the minimum between the personal swing $\ltimes_T^{(l,m)}$ of the new positions, and the next occurrence of the corresponding character after the top swings $t, t'$. 
More specifically, 
$\tswing{P \, c} = \min \{ \ltimes_T^{(l,m)}, \nextc_X(c,t)\}$ and $\tswing{P' \, c} = \min \{ \ltimes_T^{(l,m)}, \nextc_X(c,t')\}$. Since the first component of the minimum is the same, it suffices to prove that $\nextc_X(c,t) \le \nextc_X(c,t')$ to conclude $\tswing{P \, c} \le \tswing{P'\, c}$. Indeed, since $t < t'$ are positions in the same string, the next occurrence of a given character after $t$ cannot be strictly bigger than the next occurrence of the same character after $t'$. Thus, we have proved the claim for top swings; bottom swings are symmetrical. 
\end{proof}

The next corollary shows that the opposite implication holds for strict inequalities:
\begin{corollary}
\label{corollary:monotone-swings}
    Under the same hypotheses of Lemma~\ref{lemma:monotone-swings}, we have 
    \begin{itemize}
        \item $\tswing{P\, c} < \tswing{P' \, c} \Rightarrow  \tswing{P} < \tswing{P'}$
        \item $\bswing{P\, c} < \bswing{P'\, c)} \Rightarrow  \bswing{P} < \bswing{P'}$
    \end{itemize}
\end{corollary}
\begin{proof}
We reverse the proof of Lemma~\ref{lemma:monotone-swings}. Consider top swings, and recall 
$\tswing{P\, c} = \min \{ \ltimes_T^{(l,m)}, \nextc_X(c,t)\}$ and $\tswing{P'\, c} = \min \{ \ltimes_T^{(l,m)}, \nextc_X(c,t')\}$, where $t= \tswing{P}$ and $t' = \tswing{P'}$. Since the first part of the minimum is the same, and $\tswing{P\, c} < \tswing{P'\, c}$, we have two options
\begin{enumerate}
    \item $\nextc_X(c,t) \le \ltimes_T^{(l,m)} \le \nextc_X(c,t')$, where at most one inequality can be an equality; or
    \item $\nextc_X(c,t) < \nextc_X(c,t') \le \ltimes_T^{(l,m)}$.
\end{enumerate}
In any case, we have $\nextc_X(c,t) < \nextc_X(c,t')$. Since $t $ and $t'$ are positions in the same string, this relationship between the next occurrence of the same character immediately also implies $t < t'$, which concludes the proof. 
\end{proof}

We are now ready to prove our main result:
\begin{theorem}
\label{thm:inverted-swings}
For any two nodes $u\neq u'$ of \MCSDAG, let $\nodelab{u} = \recID{l}{m}{t}{b}$ and $\nodelab{u'} = \recID{l}{m}{t'}{b'}$. Then swing pairs for the same $l,m$ do not dominate each other; namely, if $t > t'$, then $b \le b'$. 
\end{theorem}
\begin{proof}
Consider any 
$v \in N^-(u)$, and $v' \in N^-(u')$. Let $\nodelab{v} = \recID{x}{y}{t_v}{b_v}$ and $\nodelab{v'} = \recID{x'}{y'}{t_{v'}}{b_{v'}}$. 
Let us first assume that $x \neq x'$ or $y \neq y'$, i.e. they are not the same pair of positions. If we look at positions $(x, y)$ and $(x',y')$ we must have either $x\le x'$ and $y >y'$, or $x > x'$ and $y \le y'$. Indeed, assume by contradiction that $x\le x'$ and $y \le y'$. Since both of these nodes have a valid extension corresponding to positions $(l,m)$, we must also have $x\le x' < l < t_v,t_{v'}$ and $y \le y' < m < b_v,b_{v'}$. Then, $(l,m)$ would not be a valid extension for $v$: the corresponding prefix is not maximal until the positions given by the swings, since we have an insertion corresponding to the character occurring at positions $(x',y')$.

We now show that $t > t'$ implies $y > y'$. By Remark~\ref{remark:remapping}, $t$ is the smallest value such that a match occurs between $X[l+1,t)$ and $Y[y, m-1]$. That is,  there is no $\tau < t$ such that $X[l+1,\tau )$ and $Y[y, m-1]$ have a match. 
Let $t>t'$, and assume by contradiction that $y \le y'$. Then, $Y[y', m-1]\subseteq Y[y, m-1]$. By definition of $t'$, there is a match between $X[l+1,t')$ and $Y[y', m-1]\subseteq Y[y, m-1]$. This is a contradiction on the minimality of $t$: there is a smaller $\tau = t' < t$ which yields a match. 
Now, since we cannot have both $y' \le y$ and $x'\le x$, we must have $x\le x'$. By a symmetrical reasoning, we show that the bottom swings must satisfy $b ' \le b$. Indeed, recall that $b$ is the minimum value for which a match occurs between $X[x, l-1]$ and $Y[m+1,b]$, and assume by contradiction that $b > b'$. Since we have $X[x',l-1] \subseteq X[x, l-1]$, we have a match between $X[x, l-1]$ and $Y[m+1,b']$ for a smaller value $b' < b$: contradiction. 

Let us now consider the case where the in-neighbors $v$ and $v'$ have the same pair of positions in their $ID$s: $x=x'$ and $y= y'$. Let us inductively consider $w_{i+1} \in N^-(w_i)$ and $w_{i+1}' \in N^-(w_i')$, where $w_0 = v$ and $w'_0 = v'$. We stop at the first $j$ such that $\nodelab{w_j} = \recID{x_j}{y_j}{t_j}{b_j}$ and $\nodelab{w'} = \recID{x'_j}{y'_j}{t'_{j}}{b'_{j}}$ with $x_j \neq x'_j$ or $y_j \neq y'_j$. Such pair satisfies the conditions of the first part of the proof, since $x_{j+1} = x'_{j+1}$ and $y_{j+1} = y'_{j+1}$ by hypothesis. Nodes $w_j$ and $w_j'$ are obtained by going backwards in the DAG for $j$ steps, starting respectively from nodes $u$ and $u'$. Since $j$ is the first index such that the corresponding positions for extensions differ, we have $\edgelab{w_{k+1}, w_{k}} = \edgelab{w'_{k+1}, w'_{k}}$ for all $k = 1,...,j-1$. By iterating Corollary~\ref{corollary:monotone-swings}, we thus have that $t > t'$ implies $t_j > t_j'$. By the first part of the proof, we therefore have $b_j \le b_j'$. By Lemma~\ref{lemma:monotone-swings}, this propagates to the end of the path in the MCS DAG, to also yield $b \le b'$.
\end{proof}


From Theorem~\ref{thm:inverted-swings}, we can derive the following result, which proves that the number of swings for a fixed pair of positions $(l,m)$ is linear: 
\begin{corollary}
\label{cor:linear-swings}
    For any given choice of $l,m$, there are just 
    $O(n)$ nodes of $\MCSDAG(X,Y)$ having the form $\nodelab{} = \recID{l}{m}{\cdot}{\cdot}$.
\end{corollary}
\begin{proof}
Let us fix  $l,m$, and consider the set $S_{l,m}\subseteq \{0,...,n-1\}\times \{0,...,n-1\}$, where $(a,b) \in S_{l,m}$ if and only there exists $u$ such that $\nodelab{u} = \recID{l}{m}{a}{b}$. By
Theorem~\ref{thm:inverted-swings}, if two pairs $(a,b)$ and $(c,d)$ belong to $S_{l,m}$, then it cannot be $a \leq c$ and $c \leq d$ (or vice versa). So one pair cannot dominate the other.

We observe that the size of $|S_{l,m}|$ is the size of the classical Pareto frontier: for an arbitrary set of points in the $\{0,...,n-1\}\times \{0,...,n-1\}$ grid, the number of points in a Pareto frontier is less than $2n$. The observation is folklore:  each point in the frontier either increases the $x$-coordinate or decreases the $y$-coordinate (possibly both). Hence, there cannot be more points on the frontiers as the sum of the $n$ possible $x$-coordinates plus the $n$ possible $y$-coordinates. Hence,  $|S_{l,m}| = O(n)$.
\end{proof}

Therefore, the number of nodes of the MCS DAG is cubic, as we have $O(n^2)$ choices for $l,m$, and every such choice gives at most a linear amount of swings $(t,b)$. 
Furthermore, it is immediate by construction of $\MCSDAG$ that the out-degree of every node is at most $\sigma$ (the characters that are valid extensions the given prefix). Therefore, we proved the space occupancy of $O(n^3 \sigma)$ memory words stated in Theorem~\ref{thm:main-CAT}.


\section{Efficient Operations on \MCSDAG}
\label{sec:operations}

We describe here how to support some operations on \textsf{compact}  $\MCSDAG(X,Y)$, which has source $s$, target $t$, and no unary nodes. We assume that each node $u$ stores the number $p(u)$ of $ut$-paths. As \textsf{compact}  $\MCSDAG(X,Y)$ is a DAG of $O(n^3 \sigma)$ edges, we can computed $p(u)$ for each node $u$ in total $O(n^3 \sigma)$ time by running a DFS, as $p(u)$ is the sum of the $p(v)$'s for the out-neighbors $v$'s of $u$.

\subsection{CAT Enumeration of MCSs}
\label{sub:cat-stpaths}
The strings in $\MCS(X,Y)$ can be listed in lexicographic order by enumerating the (labeled) $st$-paths in \textsf{compact}  $\MCSDAG$, which can be done in Constant Amortized Time with a simple DFS algorithm where the out-neighbors of each node are visited in increasing order of the labels of their outgoing edges. Although 
folklore, for completeness we sketch the DFS algorithm here. 


Consider a node $u$, where initially $u=s$, and denote the set of all $ut$-paths in $G=$ \textsf{compact}  $\MCSDAG(X,Y)$ by $\stpaths{u}{t}{G}$. The central idea is that any $ut$-path starts with $u$, followed by an element of $\stpaths{v}{t}{G\setminus u}$, i.e., a path in $G\setminus u$ (i.e.~$u$ and its incident edges removed) from an out-neighbor $v$ of $u$ to $t$. Since $G$ is a DAG, we can go a step further: a path from $u$ cannot reach $u$ again, therefore it is \textit{not} even necessary to remove $u$ (and its incident edges) from the DAG. We can thus represent $\stpaths{u}{t}{G}$ as the following disjoint union:
$
    \stpaths{u}{t}{G} = \bigcup_{v \in N(u)} \{\edgelab{u,v} P  \mid P \in \stpaths{v}{t}{G}\}.
$
From this, it is immediate that enumeration can be performed by keeping a current path prefix which we expand by traversing the DAG, backtracking every time computation terminates for all children of a node.



To obtain constant amortized time from this, we only need three simple observations:
\begin{itemize}
    \item Every path leads to $t$, so each branch of the computation leads to a leaf of the recursion tree representing a solution.
    \item Each node except $t$ has at least $2$ out-neighbors due to the path compression in $G$, so the recursion tree generated has no unary nodes, meaning the number of internal recursion-nodes is not greater than the leaves (solutions).
    \item In each recursive node we spend just constant time per recursive child, so the total time is $O(N)$, where $N$ is the number of solutions.
\end{itemize}

\subsection{Searching, Selecting, and Ranking}

We observe that each node $u$ has at most $\sigma$ out-neighbors and the edges towards them are distinct (each must start with a different character of $\Sigma$): this constitutes a ``lexicographical partition'' of the paths from $u$ to $t$ since, after a common prefix, a path starting with a larger character will lead to a lexicographically larger string. 

Searching a string $P$ as a prefix of strings from $\MCS(X,Y)$  traverses \textsf{compact} $\MCSDAG$ starting from $s$ and matching the characters in $P$ along the labels for the edges in the path. Either the search fails before reaching the end of $P$, or it succeeds and leads to a node $u$. At this point, we can run the CAT enumeration (Section~\ref{sub:cat-stpaths}) starting from $u$ to list all the $ut$-paths and so all the extensions of $P$ to strings in $\MCS(X,Y)$.

Selecting the $i$th string in lexicographic order from $\MCS(X,Y)$ is similar but uses the information $p(v)$ in each traversed node $v$. It scans the outgoing edges in order of their labels, and checks the corresponding $p(v)$: whenever the edges scanned have the smallest partial sum $\geq i$, it follows the current edge and subtracts from $i$ this partial sum. It stops at node $t$. Let $S$ be the string thus found by concatenating the labels on the traced path from $s$ to $t$. The number of traversed nodes is at most $|S|+1$, and in each node we will consider up to $\sigma$ edges, for a total cost of $O(|S|\log \sigma)$, as this can be further optimized by storing, for each edge, the sum of the annotations of the previous edges, and using binary search on these $O(\sigma)$ sums. Ranking $S$ is like searching $S$ and using the partial sums as mentioned above to get a total sum of $i$ when $t$ is reached.

\begin{theorem}
\label{thm:main-CAT}
Given two strings $X$ and $Y$ of length $n$ on an alphabet of size $\sigma$, \textsf{compact} $\MCSDAG(X,Y)$ stores $\MCS(X,Y)$ in space $O(n^3\sigma)$ and can be built in $O(n^3\sigma\log n)$ time. It supports the following operations:
\begin{itemize}
    \item For a given string $P$, report all the strings with prefix $P$ from $\MCS(X,Y)$ in $O(|P| \log \sigma + \mathit{occ})$ time, where $\mathit{occ}$ is the number of reported strings.
    \item For any integer $1 \leq i \leq |\MCS(X,Y)|$,  select the $i$th string $S$ in lexicographic order from $\MCS(X,Y)$, in $O(|S| \log \sigma)$ time.
    \item For any string $S \in \MCS(X,Y)$, return its rank $i$
    among the strings  in lexicographic order from $\MCS(X,Y)$, in $O(|S| \log \sigma)$ time.
    \item List all strings in $\MCS(X,Y)$ in $O(|\MCS(X,Y)|)$ time, i.e., Constant Amortized Time.
\end{itemize}
\end{theorem}

It is worth noting that (\textsf{compact}) $\MCSDAG$ can be stored in a succinct way, for example, using some recent methods introduced for automata~\cite{ChakrabortyGSS23}.

\section{Conclusions}

In this paper we considered the problem of storing and searching the Maximal Common Subsequences of two input strings, as they may reveal further common structures in sequence analysis with respect to the LCSs. Our main contribution is that of reducing time and space from exponential bounds, using the current state of the art, to polynomial bounds, using our new compact DAG, which supports CAT
enumeration, counting, and random access to the $i$-th element (i.e., rank and select operations) in nearly optimal time.

\newpage
\bibliography{arxiv.bib}

\end{document}